\documentclass[12pt]{article}

\usepackage{sbc-template}
\usepackage{graphicx}
\usepackage{url}
\usepackage{newtxtext}
\usepackage[utf8]{inputenc}
\usepackage[english]{babel}
\usepackage{amsmath}
\usepackage{amsthm}
\usepackage{amssymb}
\usepackage{xspace}
\usepackage{colortbl}
\usepackage{cite}
\usepackage{complexity}
\usepackage{thmtools}
\usepackage{bbm}
\usepackage{bm}
\usepackage{indentfirst}
\usepackage[subrefformat=parens]{subcaption}
\usepackage{tikzset}
\usepackage[onelanguage,english,ruled,linesnumbered]{algorithm2e}
\usepackage[noabbrev,nameinlink]{cleveref}
\usepackage{placeins}

\SetAlCapFnt{\footnotesize}

\SetCommentSty{mycommfont}

\DeclareMathOperator{\Net}{\mathcal{N}\xspace}

\newcommand{\nao}[1]{\overline{#1}}

\newcommand{\problem}[3]{
    \vspace{4pt}
    \arrayrulecolor{gray}
    \noindent
    \begin{tikzpicture}[node distance=14pt]
        \node (tb) [inner sep=4pt, rounded corners=2pt, fill=gray!20, draw=gray] {
            \begin{tabular}{p{0.1\textwidth}p{0.825\textwidth}}
                \textbf{Problem:}  & \textbf{\textsc{#1}} \\ \hline
                \textbf{Input:}    & #2 \\
                \textbf{Question:} & #3
            \end{tabular}
        };
    \end{tikzpicture}
    \arrayrulecolor{black}
    \vspace{4pt}
}

\newtheorem{proposition}{Proposition}[section]
\newtheorem{theorem}[proposition]{Theorem}
\newtheorem{lemma}[proposition]{Lemma}
\newtheorem{corollary}[proposition]{Corollary}

\declaretheoremstyle[%
    spaceabove=6pt,%
    spacebelow=6pt,%
    headfont=\normalfont\itshape,%
    postheadspace=1em,%
    qed=\qedsymbol%
]{shortproof-style}

\declaretheorem[name={Proof},style=shortproof-style,unnumbered]{shortproof}

\crefname{lemma}{Lemma}{Lemmas}
\crefname{theorem}{Theorem}{Theorems}
\crefname{corollary}{Corollary}{Corollaries}
\crefname{conjecture}{Conjecture}{Conjectures}
\crefname{definition}{Definition}{Definitions}
\crefname{property}{Property}{Properties}
\crefname{section}{Section}{Sections}
\crefname{subseciton}{Subsection}{Subsections}
\crefname{equation}{Equation}{Equations}
\crefname{figure}{Figure}{Figures}
\crefname{algorithm}{Algorithm}{Algorithms}
\crefname{table}{Table}{Tables}

\title{Minimum cost flow decomposition on arc-coloured networks}

\author{Cláudio Soares de Carvalho Neto\inst{1}, Ana Karolinna Maia\inst{1},\\
        Cláudia Linhares Sales\inst{1}, Jonas Costa Ferreira da Silva\inst{2}}

\address{
    Universidade Federal do Ceará\\
    Fortaleza-CE, Brasil.
    \email{claudio@lia.ufc.br, karolmaia@ufc.br, linhares@dc.ufc.br}
    \nextinstitute
    Universidade Federal do Amazonas\\
    Manaus-AM, Brasil
    \email{jonas.costa@icomp.ufam.edu.br}
}

\begin{document}
    \maketitle

    \begin{abstract}
        A \emph{network $\Net$} is formed by a (multi)digraph $D$ together with a \emph{capacity function} $u : A(D) \to \mathbbm{R}_+$, and it is denoted by $\Net = (D,u)$. A  \emph{flow on $\Net$} is a function $x: A(D) \to \mathbbm{R}_+$ such that $x(a) \leq u(a)$ for all $a \in A(D)$, and it is said to be $k$-splittable if it can be decomposed into up to $k$ paths \cite{BaierKSF}. We say that a flow is $\lambda$-uniform if its value on each arc of the network with positive flow value is exactly $\lambda$, for some $\lambda \in \mathbbm{R}_+^*$.

        Arc-coloured networks are used to model qualitative differences among different regions through which the flow will be sent \cite{Granata2013}. They have applications in several areas such as communication networks, multimodal transportation, molecular biology, packing etc. 

        We consider the problem of decomposing a flow over an arc-coloured network with minimum cost, that is, with minimum sum of the cost of its paths, where the cost of each path is given by its number of colours. We show that this problem is $\NP$-Hard for general flows. When we restrict the problem to $\lambda$-uniform flows, we show that it can be solved in polynomial time for networks with at most two colours, and it is $\NP$-Hard for general networks with three colours and for acyclic networks with at least five colours.
    \end{abstract}

    \section{Introduction}
        Flows in networks are one of the most important tools to solve problems in graphs and digraphs. They constitute a generalization of some classical problems as shortest path and those related to finding internally arc-disjoint paths from a vertex to another. They 
        are widely studied as they allow, with a certain elegance and simplicity, modeling problems in different areas of study such as transportation, logistics and telecommunications. A long list of results related to flows can be found in \cite{FordFulkerson,Ahuja}. The simple theory combined with its applicability to real-life problems makes flows a very attractive topic to study.

        We use flows to model situations in which we need to represent some commodity moving from one part to the other of the network. The \emph{multi-commodity flow} problem asks for a flow that satisfies the demands for each commodity between a source and a destination, respecting the capacities of the arcs. The restriction to this problem where the demand for each commodity must be sent along a single path was proposed by \cite{Kleinberg96} and it is called the \emph{unsplittable flow} problem. The author showed that this problem contains a wide range of $\NP$-Complete problems, such as partitioning, scheduling, packing etc.

        A natural generalization of the unsplittable flow problem is to allow the demand of each commodity to be sent through a limited number of paths. This version is called \emph{splittable flow} problem and was introduced by \cite{BaierKSF}. Such problems arise, for example, in communication networks, where clients may demand connections with specific bandwidths between pairs of nodes. If these bandwidths are too high, it might be impossible for the network administrator to satisfy them unsplittably. On the other hand, the client may not wish to deal with many connections of small bandwidths. The flow can be viewed as a collection of paths along which an amount of the commodity is sent. Thus, we say that a flow is $k$-splittable if it can be decomposed into up to $k$ such paths. Each path may represent, for instance, a container associated with a route. Determining the minimum $k$ such that a given flow is $k$-splittable implies minimizing the number of containers needed to send the commodity.

        One may also consider to embed other kinds of structural attributes to the problem. As mentioned by \cite{Granata2013}, a considerable amount of work has been spent to face problems related to arc-coloured digraphs. They are used to model situations where it is crucial to represent qualitative differences among different regions of the graph itself. Each colour represents a different property (or set of properties). They have applications in several areas such as communication networks, multimodal transportation, packing among others. For example, in communication networks, colours may be used to model risks.

        According to \cite{CoudertSRRG2007}, in many situations it is needed to consider the correlations between arcs of the network, motivated by the network survivability concept of \emph{Shared Risk Resource Group} (SRRG). A SRRG is a set of resources that will break down simultaneously if a given failure occurs. One can then define the “safest” path as the one using the least number of groups. They are modelled by associating to each risk a colour, and to each resource an arc coloured by each of the colours representing the risks affecting it.


        The problem we propose in this work join aspects of the splittable flow problem and arc-coloured networks. For a given network and a flow in it, we want to decompose the flow with minimum cost, that is, with minimum sum of the cost of its paths, where the cost of each path is given by its number of distinct colours.

        For monochromatic networks, the problem consists of minimising the number of paths in the decomposition. In \cite{Hartman}, the authors showed that this problem is $\NP$-Hard for networks with three distinct flow values on the arcs and can be solved in polynomial time for networks with two distinct values and the smaller one divides the other. In this paper, we show that this problem can be solved in polynomial time for any two distinct flow values in acyclic networks.

        For bichromatic networks, we show that the problem can be solved in polynomial time when the flow is $\lambda$-uniform or when each colour is associated with a flow value and the smaller one divides the other. Unlike the monochromatic version, the problem remains open for the general case with two flow values.

        For networks with three colours, we prove that the problem is $\NP$-Hard even when the flow is uniform and the degree of each vertex, except for the source and sink, is at most $6$. When the problem is restricted to acyclic networks with at least five colours and uniform flows, it remains $\NP$-Hard. Therefore, the problem is difficult to solve for a small number of colours, even for simpler networks and flows.

        In \cref{sec:definitions} we give the basic notations and terminologies of the theory of network flows and formalize the proposed problem. In \cref{sec:complexity}, we make a study about the complexity of the problem for general networks and flows, and for some restricted cases.

    \section{Definitions and Terminology}\label{sec:definitions}
        We assume that the reader is familiar with the basic concepts in graph theory, specially with the notations for digraphs as in \cite{BondyMurty,BangJensen}.

        We denote by $D=(V,A,c)$ an arc-coloured (multi)digraph wit h vertex set $V$, arc set $A$ and an arc colouring $c: A(D) \rightarrow \{1, \ldots, p\}$. The colouring does not need to be proper, that is, two adjacent arcs may have the same colour. If every arc has the same colour, we omit the letter $c$ from the notation. We denote by $n_c(D)$ the number of distinct colours of a digraph $D$, by \emph{colours(D)} the image of $c$ and by \emph{span(i,D)} the subdigraph of $D$ induced by the arcs with colour $i$. The cardinalities of the sets $V$ and $A$ are referred respectively by $n$ and $m$.

        A \emph{network} $\Net$ is formed by a (multi)digraph $D = (V,A,c)$ with a \emph{capacity function} $u : A(D) \rightarrow \mathbbm{R}_+$, and it is denoted by $\Net = (D,u)$. We use the notation $\Net = (D,u \equiv \lambda)$ to say that $u(ij) = \lambda$, for every $ij \in A(D)$. For convenience, we will show the notation for digraphs, but it can be easily generalised to multidigraphs. A \emph{flow} in $\Net$ is a function $x: A(D) \to \mathbbm{R}_+$ such that $x(ij) \leq u(ij)$ for every arc $ij \in A(D)$. For the sake of simplicity, we may use $x_{ij}$, $u_{ij}$ and $c_{ij}$ to denote $x(ij)$, $u(ij)$ and $c(ij)$, respectively. We say that $x$ is an \emph{integer flow} if $x_{ij} \in \mathbbm{Z}_+$ for every arc $ij \in A(D)$. For some positive integer $\lambda$, we say that a flow $x$ is $\lambda$\emph{-uniform} if $x_{ij} \in \{0, \lambda\}$ for every arc $ij \in A(D)$. The \emph{support} of a network $\Net=(D,u)$ with respect to a flow $x$ is the digraph induced by the arcs of $D$ with positive flow value.

        With respect to a flow $x$ in a network $\Net = (D,u)$, for a vertex $v \in V(D)$, we define $x^+(v) = \sum_{vw \in A} x(vw)$ and $x^-(v) = \sum_{uv \in A} x(uv)$, that is, $x^+(v)$ is the amount of flow leaving $v$ and $x^-(v)$ is the amount of flow entering $v$. The \emph{balance} of a vertex $v$ in $x$ is defined by $b_x(v) = x^+(v)-x^-(v)$.

        Let $s$, $t$ be distinct vertices of a network $\Net = (D,u)$. An $(s,t)$-flow is a flow $x$ in which $b_x(s) = -b_x(t) = k$, for some $k \in \mathbbm{R}_+$, and $b_x(v) = 0$ for every vertex other than $s$ and $t$ (this is called \emph{Conservation Condition}). The value of such a flow is defined by $|x| = b_x(s)$. Usually, it is convenient to see an $(s,t)$-flow as a collection of paths from $s$ to $t$, where each of these paths is associated to an amount of flow that it carries.

        A \emph{path flow} in a network $\Net$ is a flow $x$ along a path $P$ such that $x_{ij} = r$ for every $ij \in A(P)$, for some positive value $r$. Analogously, we define a \emph{cycle flow} along a cycle $C$. A \emph{circulation} is a set of cycle flows. A classical result in the Network Flows Theory states the following:

        \vspace{10pt}
        \begin{theorem}[Flow Decomposition Theorem \cite{FordFulkerson}]\label{the:flowdec}
            Every $(s,t)$-flow $x$ in a network $\Net$ can be decomposed into at most $n+m$ path flows or cycle flows in $O(mn)$ time.
        \end{theorem}
        \vspace{6pt}

        Let $\Net$ be an arc-coloured network. Suppose that each colour is associated to a risk. Then, finding a path with a minimum number of colours corresponds to finding a safer path. In \cite{YuanMCPP}, the authors showed that this problem is $\NP$-Hard.

        As defined in \cite{BaierKSF}, an $(s,t)$-flow $x$ on a network $\Net$ is said to be $k$-splittable if it can be specified by $k$ pairs $(P_1,f_1), \ldots, (P_k,f_k)$, each one representing an $(s,t)$-path $P_i$ associated to a flow value $f_i$, for $1 \leq i \leq k$, that is, if it can be decomposed into up to $k$ path flows $x^1, \ldots, x^k$, such that $\sum_{i=1}^k |x^i| = |x|$. The paths do not need to be distinct. For each path extracted in the decomposition of \cref{the:flowdec}, the flow on at least one arc becomes zero. So, apart from the circulation, every $(s,t)$-flow is $m$-splittable. Therefore, we are interested in values of $k$ less then $m$.

        The problem we propose, referred as \textsc{MinCostCFD} (Min. Cost Coloured Flow Decomposition), consists of: given a network $\Net=(D,u)$, with $D=(V,A,c)$, and an $(s,t)$-flow $x$ on it, finding a decomposition of $x$ into $\ell$ path flows, $x^1, \ldots, x^\ell$, where each $x^i$ is sent along a directed path $P_i$, while minimizing the cost of the solution given by:
        \begin{equation}\label{eq:dfcol}
            n_c(P_1, \ldots ,P_\ell) = \sum_{i=1}^{\ell} n_c(P_i)
        \end{equation}

        Observe that, together with the path flows obtained on the decomposition as above, we may also have a circulation. However, it does not affect the cost of the decomposition.

        In a $\lambda$-uniform $(s,t)$-flow $x$, the number of paths of an optimal decomposition is $|x|/\lambda$. Even though the circulation can be removed from a flow without changing its value, it plays an important role on the cost of a decomposition. In the \cref{fig:mincolfd}, there is a $\lambda$-uniform flow $x$ of value $2\lambda$ on an arc-coloured network. The label on each arc represents its colour. If we remove the cycle $abcda$, there are two path flows along the bichromatic paths $sat$ and $sct$, and so the cost is $4$. We can get a decomposition of cost $2$ by taking two path flows along the monochromatic paths $sabct$ and $scdat$. This means that removing the circulation is not a good strategy for solving the problem.

                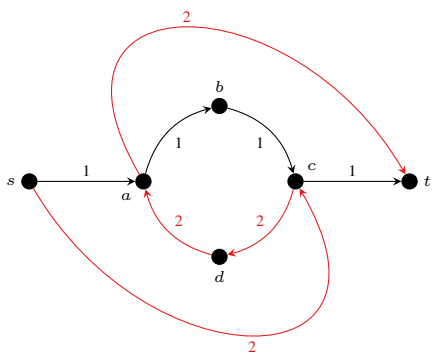
\begin{figure}[!ht]
            \centering
            \begin{tikzpicture}[>=stealth, auto=left, x=1.0cm, y=1.0cm]
                \tiny
                \clip (0.0, -0.3) rectangle (6, 4.3);
                \node [nodeb] (s) [label=180:$s$] at (0.5, 2.0) {};
                \node [nodeb] (a) [label=225:$a$] at (2.0, 2.0) {};
                \node [nodeb] (b) [label= 90:$b$] at (3.0, 3.0) {};
                \node [nodeb] (c) [label= 45:$c$] at (4.0, 2.0) {};
                \node [nodeb] (d) [label=270:$d$] at (3.0, 1.0) {};
                \node [nodeb] (t) [label=  0:$t$] at (5.5, 2.0) {};

                \path [->]
                    (s) edge [c09]                                       node{1} (a)
                        edge [c02, out=300, in=300, looseness=2.2, swap] node{2} (c)
                    (a) edge [c09, bend left, swap]                      node{1} (b)
                        edge [c02, out=120, in=120, looseness=2.2]       node{2} (t)
                    (b) edge [c09, bend left, swap]                      node{1} (c)
                    (c) edge [c02, bend left, swap]                      node{2} (d)
                        edge [c09]                                       node{1} (t)
                    (d) edge [c02, bend left, swap]                      node{2} (a);

            \end{tikzpicture}
            \caption{Influence of cycles on a decomposition of a $\lambda$-uniform flow.}
            \label{fig:mincolfd}
        \end{figure}

    \section{Complexity Results of \textsc{MinCostCFD}}\label{sec:complexity}

        In this section, we show that the \textsc{MinCostCFD} problem is $\NP$-Hard, for general networks and flows, and we show some complexity results to this problem when the number of colours and flow values on the arcs of the network are limited.

        We show the $\NP$-hardness of \textsc{MinCostCFD} by showing that its decision version, \textsc{KCostCFD}, is $\NP$-compete.
        We show a reduction from \textsc{3-Partition}. These problems are defined as follows:

        \problem{KCostCFD}
                {A network $\Net = (D,u)$, where $D=(V,A,c)$, an $(s,t)$-flow $x$ and $k \in \mathbbm{Z}_+^*$.}
                {Does $x$ admit a decomposition in $(s,t)$-path flows with cost at most $k$?}

        \vspace{-12pt}

        \problem{3-Partition}
                {A finite set $S=\{a_1, \ldots, a_{3r}\}$, a bound $T \in \mathbbm{N}$, and a value $v(a) \in \mathbbm{N}$,
                 such that $T/4 < v(a) < T/2$, for each $a \in S$, and $\sum_{i=1}^{3r} v(a_i) = rT$.}
                {Can $S$ be partitioned into $r$ subsets $S_1, \ldots, S_r$ such that, for $1 \leq i \leq r$, $\sum_{a \in S_i} v(a) = T$?}

        \vspace{-12pt}

        According to \cite{GareyJohson}, the \textsc{3-Partition} problem is strongly $\NP$-Complete. As observed by the authors, the constraints on the item values imply that every subset $S_i$ must have exactly three elements.

        \vspace{10pt}
        \begin{theorem}\label{the:kcostcfdnpc}
            The \textsc{KCostCFD} problem is $\NP$-Complete.
        \end{theorem}

        \begin{shortproof}
            First we show that \textsc{KCostCFD} is in $\NP$. For a given $(s,t)$-flow $x$ and a sequence of $(s,t)$-path flows $x^1, \ldots, x^\ell$, one can verify in polynomial time if $\textstyle \sum_{i=1}^{\ell} x_a^i = x_a$, for every arc $a$, where each $x^i$ is sent through a path $P_i$, and if $\textstyle \sum_{i=1}^{\ell} n_c(P_i) \leq k$.

            Now we show a reduction from \textsc{3-Partition} to \textsc{KCostCFD} problem. For a given set $S=\{a_1, \ldots, a_{3r}\}$ and a bound $T$, where $T/4 < v(a_i) < T/2$ and $\sum_{i=1}^{3r} v(a_i) = rT$, instance of \textsc{3-Partition}, we define $k=6r$ and build a network $\Net=(D,u)$, with $D=(V,A,c)$, and define an $(s,t)$-flow $x$ in it with value $rT$ as follows:

            \begin{itemize}
                \item $V = \{a_1, \ldots, a_{3r}, b_1, \ldots, b_r, q, s, t\}$;
                \item $A = \{sa_i, a_iq \mid 1 \leq i \leq 3r\} \cup \{qb_j, b_jt \mid 1 \leq j \leq r\}$;
                \item we define $u_{sa_i} = u_{a_iq} = v(a_i)$, e $u_{qb_j} = u_{b_jt} = T$;
                \item we set $x_{ij} = u_{ij}$ for every $ij \in A$ (this is possible, because $\textstyle \sum_{i=1}^{3r} v(a_i) = rT$);
                \item finally, we define the colouring $c_{sa_i} = c_{a_iq} = r+1$ e $c_{qb_j} = c_{b_jt} = j$.
            \end{itemize}

            The network described above has $4r+3$ vertices and $8r$ arcs. So, the construction is done in polynomial time in the input size and it is illustrated in \cref{fig:red3partkcostcfd} (the label on each arc indicates its colour).

                    \begin{figure}[!ht]
            \centering
            \begin{tikzpicture}[>=stealth, auto=left]
                \tiny
                \node [nodeb] (s)  [label=180:$s$]      at (0.0, 2.00) {};
                \node [nodeb] (a1) [label= 90:$a_1$]    at (2.0, 4.00) {};
                \node [nodeb] (a2) [label= 90:$a_2$]    at (2.0, 3.00) {};
                \node [nodeb] (a3) [label= 90:$a_3$]    at (2.0, 2.00) {};
                \node         (r1)                      at (2.0, 1.40) {\normalsize $\vdots$};
                \node [nodeb] (a4) [label=270:$a_{3r}$] at (2.0, 0.50) {};
                \node [nodeb] (q)  [label=  0:$q$]      at (4.0, 2.00) {};
                \node [nodeb] (b1) [label= 90:$b_1$]    at (6.0, 2.80) {};
                \node         (r2)                      at (6.0, 2.10) {\normalsize $\vdots$};
                \node [nodeb] (b2) [label=270:$b_r$]    at (6.0, 1.20) {};
                \node [nodeb] (t)  [label=  0:$t$]      at (8.0, 2.00) {};

                \path [->]
                    (s)  edge [bend left=30]          node {r+1} (a1)
                         edge [bend left=20, pos=0.7] node {r+1} (a2)
                         edge                         node {r+1} (a3)
                         edge [bend right=30, swap]   node {r+1} (a4)
                    (a1) edge [bend left=30]          node {r+1} (q)
                    (a2) edge [bend left=20, pos=0.3] node {r+1} (q)
                    (a3) edge                         node {r+1} (q)
                    (a4) edge [bend right=30, swap]   node {r+1} (q)
                    (q)  edge [bend left=20]          node {1}   (b1)
                         edge [bend right=20, swap]   node {r}   (b2)
                    (b1) edge [bend left=20]          node {1}   (t)
                    (b2) edge [bend right=20, swap]   node {r}   (t);
            \end{tikzpicture}
            \caption{Reduction from \textsc{3-Partition} to \textsc{kCostCFD}.}
            \label{fig:red3partkcostcfd}
        \end{figure}
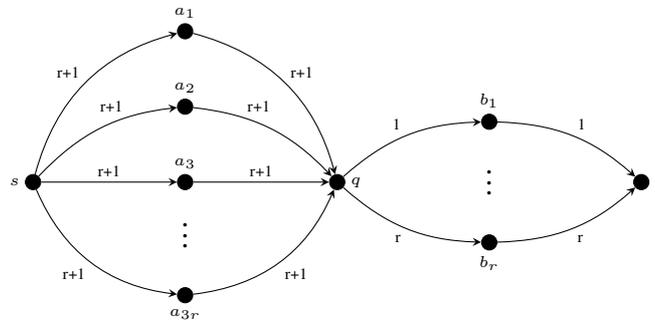

            Now we show that the answer to \textsc{3-Partition} is yes if and only if the $(s,t)$-flow $x$ admits a decomposition into path flows with cost at most $k$.

            If the answer to \textsc{3-Partition} is yes, the flow can be decomposed into $3r$ path flows of cost $2$, each, as follows: let $\{S_1, \ldots, S_r\}$ be a solution to \textsc{3-Partition} problem. For each $a_i \in S_j$, with $1 \leq j \leq r$, there exists a path flow through $sa_iqb_jt$ with value $v(a_i)$, using two colours, $j$ and $r+1$, in the network. Thus, there exists a decomposition of the flow $x$ with cost $6r$.

            Now assume that the $(s,t)$-flow $x$ defined in the network $\Net$ above can be decomposed into path flows with total cost at most $k=6r$. In fact, this cost cannot be less than $6r$, as there are exactly $3r$ arc-disjoint paths from $s$ to $q$, and every path from $s$ to $t$ has two colours. Therefore, $3r$ bichromatic $(s,t)$-path flows are needed. Each vertex $a_i$, for $1 \leq i \leq 3r$, must be in exactly one path flow, which value is $v(a_i)$. Note that there are exactly three path flows through each vertex $b_j$, for $1 \leq j \leq r$, due to the restriction on the values of $a_i$ $(T/4 < v(a_i) < T/2)$. Let them be $sa_xqb_jt$, $sa_yqb_jt$ and $sa_zqb_jt$. For every three of these path flows, we construct a subset $S_j = \{a_x, a_y, a_z\}$ such that $v(a_x) + v(a_y) + v(a_z) = T$.
        \end{shortproof}

        When dealing with an $\NP$-Hard problem, a natural approach is to impose some restrictions to the input so that we can find polynomial-time algorithms for special cases of the problem at hand. That is what we do in the following subsections.

        \subsection{Monochromatic Networks}

            When all arcs on the network have the same colour, the problem consists of determining the smallest number of paths in which the flow can be decomposed. This is the \textsc{MinSplittableFlow} problem and its decision version, the \textsc{KSplittableFlow}, is defined bellow.

            \problem{KSplittableFlow}
                    {A network $\Net=(D,u)$, with $D=(V,A)$, an $(s,t)$-flow $x$ and $k \in \mathbbm{Z}_+^*$.}
                    {Is $x$ $k$-splittable?}

            \vspace{-12pt}

            According to \cite{Vatinlen2008}, the \textsc{KSplittableFlow} is $\NP$-Complete. Let us look at some cases regarding the number of distinct flow values in the arcs of the network.

            \vspace{10pt}
            \begin{lemma}[\cite{Hartman}]\label{lem:ksplitmult}
                Let $\Net$ be a network in which all capacities are multiples of $\lambda$, and let $x$ be a maximum $(s,t)$-flow in $\Net$. Then, $|x|$ is multiple of $\lambda$ and can be decomposed into exactly $|x|/\lambda$ $(s,t)$-path flows of value $\lambda$.
            \end{lemma}

            If an $(s,t)$-flow $x$ in a network $\Net$ is $\lambda$-uniform, then by \cref{lem:ksplitmult} the minimum number of paths is $|x|/\lambda$.

            Given a network $\Net=(D,u)$, with $D=(V,A)$, a flow $x$ and an integer $a$, we define the $support(\Net, x, a)$ operation that returns a network $\Net'=(D',u')$, with $D'=(V,A')$ and $A'=\{ij \in A \mid x_{ij} \geq a\}$, and $u'_{ij} = x_{ij}$ for every $ij \in A'$.

            According to \cite{Hartman}, given a network $\Net$ and an $(s,t)$-flow $x$ such that $x_{ij} \in \{a,b\}$ for every arc $ij$ in the network, and $b$ divides $a$, it is possible to find the optimal solution for the \textsc{MinSplittableFlow} problem under these conditions using the following algorithm:

            \begin{enumerate}
                \item [(i).]   Do $\Net_a = support(\Net,x,a)$ and calculate a maximum $(s,t)$-flow $x_a$
                               in $\Net_a$, which is decomposed into $p_1$ paths of value $a$;
                \item [(ii).]  Obtain a flow $x'$ in $\Net$ by decreasing the value $x_a$ from $x$ (arc by arc);
                \item [(iii).] Do $\Net_b = support(\Net,x',b)$ and calculate a maximum $(s,t)$-flow $x_b$
                               which is decomposed into $p_2$ paths of value $b$.
            \end{enumerate}

            In step $(i)$, it is possible to obtain $p_1$ using \cref{lem:ksplitmult}. In step $(iii)$, since all flow values $x_b$ are multiples of $b$, it is also possible to obtain $p_2$ using \cref{lem:ksplitmult}.

            \vspace{10pt}
            \begin{theorem}[\cite{Hartman}]\label{the:ksplit2v}
                Consider a network $\Net$ and an $(s,t)$-flow $x$ in it, such that there are only two distinct flow values $a$ and $b$ on the arcs, and $b \mid a$ ($b$ divides $a$). The solution produced by the above algorithm is optimal.
            \end{theorem}

            We considered the case in which the network has at most two distinct flow values on its arcs, let us say $a$ and $b$ and assume that $a > b$, and one is not necessarily a multiple of the other. We proposed the \cref{alg:decfluxo2v}, based on the Euclide's algorithm for computing the greatest common divisor (GCD) of two positive integers. It gets as input a network $\Net$, two vertices $s$ and $t$, an $(s,t)$-flow $x$, two positive integers $a$ and $b$, and returns an array $P$ with the number of path flows on each iteration of the decomposition of $x$.

            Notice that, after each iteration of the command \emph{while} of the \cref{alg:decfluxo2v}, the value of $a$ is reduced to at most a half of it. To see this, we must consider two cases: if $b \leq a/2$, then $a~mod~b < b \leq a/2$; otherwise, $a~mod~b = a-b < a/2$.

            The values of $a$ and $b$ are reduced, alternately, to the half of their previous values. Then, the number of iterations done by the algorithm is $O(\log a + \log b)$. The complexity of the algorithm depends on the complexity of finding a maximum $(s,t)$-flow. This can be done in polynomial-time (see some algorithms to this purpose in \cite{Ahuja}).

            \begin{algorithm}
                \footnotesize\itshape
                \caption{\footnotesize FlowDecomposition2V($\Net,s,t,x,a,b$)}
                \label{alg:decfluxo2v}
                \Begin{
                    $i \leftarrow 1$\;
                    \While{($b > 0$)}{
                        $\Net_i \leftarrow support(\Net, x, a)$\;
                        Update the capacity of each arc $ij$ in $\Net_i$ to $\lfloor x_{ij} / a \rfloor \cdot a$\;
                        Calculate a max. $(s,t)$-flow $x'$ in $\Net_i$\;
                        $P[i] \leftarrow |x'| / a$\;
                        $x \leftarrow x - x'$\tcp*[l]{Update the flow $x$ on arcs with $x'>0$}
                        $r \leftarrow a \mod~b$\;
                        $a \leftarrow b$\;
                        $b \leftarrow r$\;
                        $i \leftarrow i + 1$\;
                    }
                    $\Net_i \leftarrow support(\Net, x, a)$\;
                    Calculate a max. $(s,t)$-flow $x'$ in $\Net_i$\;
                    $P[i] \leftarrow |x'| / a$\;
                    \Return{$P$}\;
                }
            \end{algorithm}

            \FloatBarrier

            Let $P = \langle p_1, \ldots, p_m \rangle$ be the array returned by \cref{alg:decfluxo2v}, and $a_i$ the value of $a$ at the iteration $i$. Remark that $a_1 > \cdots > a_m$ and each $p_i$ is obtained from a maximum flow in a network in which the capacity of each arc is at least $a_i$. As $a_m | a_{m-1}$, the flow $x$ may be decomposed into $\textstyle \sum_{i=1}^{m} p_i$ path flows.

            To calculate $p_1$, every arc in the network has capacity $a_1 = a$. For $i > 1$, the capacity values on the arcs must be in $\{a_i, \lfloor a_{i-1} / a_i \rfloor \cdot a_i\}$, that is, there are at most two distinct flow values on the arcs and the smallest one divides the other, and there are no path flows of value greater than $a_i$. It is possible to get $p_i$ in polynomial-time (see \cite{Hartman}). Furthermore, by the conservation condition, after calculating $p_i$, all arcs with capacity greater than $a_i$ must have been used by $(s,t)$-path flows of value $a_i$.

            The decomposition produced by \cref{alg:decfluxo2v} is illustrated in \cref{fig:flow2v1c}. The label on each arc represents its colour. We have in \subref{sfg:flow2v1ca} the original network and a flow of value $35$. In \subref{sfg:flow2v1cb}, the network $N_1$ with arcs whose flow value is at least $a_1 = 7$. In \subref{sfg:flow2v1cc}, the network $N_2$ with arcs whose flow value is at least $a_2 = 5$. Note that the arcs with a flow value $7$ in $N_1$ were adjusted to $\lfloor 7 / 5 \rfloor \cdot 5 = 5$. In \subref{sfg:flow2v1cd}, the network $N_3$ with arcs whose flow value is at least $a_3 = 2$. Similarly to the previous case, the arcs with a flow value $5$ in $N_2$ were adjusted to $\lfloor 5 / 2 \rfloor \cdot 2 = 4$. Finally, in \subref{sfg:flow2v1ce}, we have the network $N_4$ with arcs whose flow value is at least $1$. The number of paths in this decomposition is given by $p_1 + p_2 + p_3 + p_4 = 0 + 5 + 4 + 2 = 11$.

            \input{fig_flow2v1c}

            The \cref{alg:decfluxo2v} also works when $a = b$. In this case, $x$ is an $a$-uniform $(s,t)$-flow and the returned array is $P = \langle p_1 \rangle$, where $p_1 = |x| / a$. If $a > b$ and $b | a$, the result produced by this algorithm is optimal. This corresponds to the result showed by \cite{Hartman} to this case. The following results are related to the case in which $a > b$ and $b \nmid a$, and the network is defined over an acyclic digraph (DAG).

            \vspace{10pt}
            \begin{lemma}\label{lem:mincostcfd2v}
                Let $\Net$ be a network defined on a $DAG$, and $a$ and $b$ the distinct flow values in its arcs, with $a > b$ and $b \nmid a$. Let $P = \langle p_1, \ldots, p_m \rangle$ be the array returned by \cref{alg:decfluxo2v}. After calculating each $p_i$, there are at most two distinct flow values on the arcs of $\Net$.
            \end{lemma}

            \begin{shortproof}
                Let $a_i$ be the value of $a$ used for calculating $p_i$ at each iteration $i$. Notice that each $p_i$ is obtained from a maximum flow in a network with a capacity at least $a_i$. Since $a > b$, note that $a_1=a$, $a_2=b$, and $a_i=a_{i-2} \mod a_{i-1}$, for $3 \leq i \leq m$, and that $a_m \mid a_{m-1}$ ($a_m$ divides $a_{m-1}$).

                We proceed by induction on the size of $P$. Initially, the flow values on the arcs are $a_1$ and $a_2$. After calculating $p_1$, the flow value in each arc is in $\{a_1, a_2\}$. Similarly, after calculating $p_2$, the flow value in each arc is in $\{a_2, a_3\}$. For $k < m$,  we suppose that, after calculating $p_k$, there are at most two flow values on the arcs of the network which are in $\{a_k, a_{k+1}\}$. We show that this also holds for the calculation of $p_{k+1}$ for flow values in $\{a_{k+1}, a_{k+2}\}$.
                
                For $k+1 < m$, suppose that there exists at least one arc $uv$ with a flow value $a_k$ after calculating $p_{k+1}$. Let $P'$ be a maximal $(s', t')$-path flow that contains $uv$ and has value $a_k$. Note that $s' \neq s$ or $t' \neq t$, as the value of $p_k$ is maximum. If $s' \neq s$, then since $b(s') = 0$, there are at least $a_k$ units of flow entering $s'$ through $(s, s')$-path flows of value $a_{k+1}$. Let $P_s'$ be one of these paths. Otherwise, if $s = s'$ then $P_s'$ is empty. Similarly, if $t' \neq t$, then since $b(t') = 0$, there are at least $a_k$ units of flow leaving $t'$ through $(t',t)$-path flows of value $a_{k+1}$. Let $P_t'$ be one of these paths. Otherwise, if $t = t'$ then $P_t'$ is empty. We can send $a_{k+1}$ units of flow along the $(s,t)$-path $P_s'P'P_t'$. This is a contradiction with the maximality of $p_{k+1}$.

                If $k+1 = m$, then $a_{k+1} \mid a_k$. All flow paths will have value $a_{k+1}$. After calculating $p_{k+1}$, the flow value in each arc is zero, and the result follows.

            \end{shortproof}

            By the proof of \cref{lem:mincostcfd2v}, when $a = a_i$, for $i > 1$, there are no $(s,t)$-paths with capacity $a_{i-1}$ in the network. Then, we have the following result:

            \vspace{10pt}
            \begin{theorem}\label{the:mincostcfd2v}
                If $\Net$ is acyclic, the solution given by the \cref{alg:decfluxo2v} is optimal.
            \end{theorem}

            \begin{shortproof}
                Let $P = \langle p_1, \ldots, p_m \rangle$, and let $a_i$ be the value of $a$ in the calculation of $p_i$, with $1 \leq i \leq m$. We know that $a_1 > a_2 > \ldots > a_{m-1} > a_m$ and besides $a_m | a_{m-1}$. Initially, there are only two distinct flow values -- $a$ and $b$ -- on the arcs of the network. After each iteration $i$, by \cref{lem:mincostcfd2v}. there are no more paths of capacity $a_{i-1}$. The number of path flows in the decomposition of $x$ given by \cref{alg:decfluxo2v} is $p_1 + \ldots + p_m$. This number is minimum, since each $p_i$ was obtained from a maximum flow of value $p_i \cdot a_i$. The decrease of one unit of any $p_i$ would imply the increase of at least $\lceil a_i/a_{i+1} \rceil \geq 2$ path flows in the subsequent iterations.
            \end{shortproof}

            The problem remains open when there are just two distinct flow values in the arcs of the network, and the digraph over which the network is defined is not acyclic.

            What if we have a network with more than two flow values in its arcs? In this case, the greedy approach of extracting path flows with higher values first does not work. We can see this in the example of \cref{fig:flow3v1c}. The label of each arc indicates its flow value. If the path flow of value $4$ is extracted, there will be $6$ path flows of value $1$ remaining, making a total of $7$ path flows. However, it is possible to initially extract $1$ path flow of value 2 (from the path flow of value $4$), $3$ path flows of value $2$ (starting from the arcs leading from $s$ with this value), and $2$ path flows of value $1$, making a total of $6$ path flows. To verify that this is minimal, note that to avoid the $6$ path flows of value $1$, it is necessary to pass $3$ path flows (two of value $1$ and one of value $2$) through at least one arc that has flow value $4$, particularly from $v_2$ to $v_3$ or from $v_4$ to $v_5$. As the out-degree of $s$ is $4$, there are at least $4$ $(s,t)$-path flows, and one of them must be split into $3$.

                        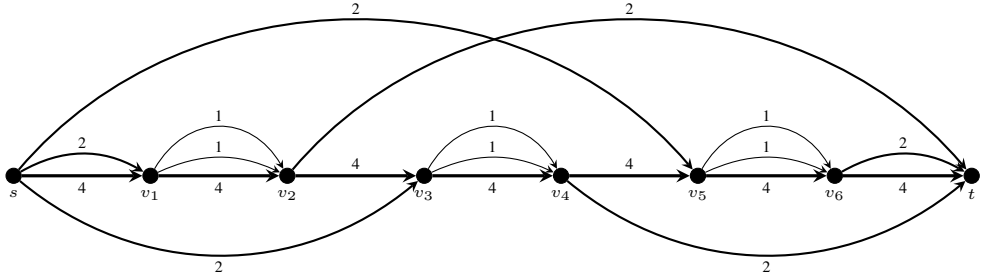
\begin{figure}[!htb]
                \centering
                \begin{tikzpicture}[>=stealth, auto=left, node distance=1.8cm]
                    \tiny
                    \node [nodeb] (s)  [label=270:$s$]                {};
                    \node [nodeb] (v1) [label=270:$v_1$, right of=s]  {};
                    \node [nodeb] (v2) [label=270:$v_2$, right of=v1] {};
                    \node [nodeb] (v3) [label=270:$v_3$, right of=v2] {};
                    \node [nodeb] (v4) [label=270:$v_4$, right of=v3] {};
                    \node [nodeb] (v5) [label=270:$v_5$, right of=v4] {};
                    \node [nodeb] (v6) [label=270:$v_6$, right of=v5] {};
                    \node [nodeb] (t)  [label=270:$t$,   right of=v6] {};

                    \path [->, very thick]
                        (s)  edge [swap] node{4} (v1)
                        (v1) edge [swap] node{4} (v2)
                        (v2) edge        node{4} (v3)
                        (v3) edge [swap] node{4} (v4)
                        (v4) edge        node{4} (v5)
                        (v5) edge [swap] node{4} (v6)
                        (v6) edge [swap] node{4} (t);

                    \path [->, thick]
                        (s)  edge [bend left=30]        node{2} (v1)
                             edge [bend right=40, swap] node{2} (v3)
                             edge [bend left=50]        node{2} (v5)
                        (v2) edge [bend left=50]        node{2} (t)
                        (v4) edge [bend right=40, swap] node{2} (t)
                        (v6) edge [bend left=30]        node{2} (t);

                    \path [->]
                        (v1) edge [bend left=25]                node{1} (v2)
                             edge [bend left=60, looseness=1.3] node{1} (v2)
                        (v3) edge [bend left=25]                node{1} (v4)
                             edge [bend left=60, looseness=1.3] node{1} (v4)
                        (v5) edge [bend left=25]                node{1} (v6)
                             edge [bend left=60, looseness=1.3] node{1} (v6);
                \end{tikzpicture}
                \caption{Network with 3 distinct flow values on the arcs.}
                \label{fig:flow3v1c}
            \end{figure}

            \FloatBarrier

            This difference between the optimal and greedy solutions can be arbitrarily large. To show this, we can generalise the network from \cref{fig:flow3v1c} to $2n+2$ vertices ($s$, $t$, $v_1, \ldots, v_{2n}$). With the greedy approach, the number of path flows is $2n+1$. If we extract the paths like in the second solution described above, we get $n+3$ path flows.

            In \cite{Hartman}, the authors showed that the \textsc{KSplittableFlow} problem is $\NP$-Complete if there are three different flow values on the arcs of the network.

            \vspace{10pt}
            \begin{theorem}[\cite{Hartman}]\label{the:ksplit3v1c}
                Let $x$ be a flow such that on each arc the flow value is either $1$, $2$ or $4$, and let $k$ be an integer. Then it is $\NP$-complete to decide if there exists a decomposition of $x$ into at most $k$ paths.
            \end{theorem}

        \subsection{Bichromatic Networks}

            Initially, we treat a simple case of the \textsc{MinCostCFD} for bichromatic networks - when the flow is $\lambda$-uniform.

            \vspace{10pt}
            \begin{theorem}\label{the:mincostcfduni}
                \textsc{MinCostCFD} can be solved in polynomial-time if the network is bichromatic and the $(s,t)$-flow $x$ is $\lambda$-uniforme.
            \end{theorem}

            \begin{shortproof}
                Since $x$ é $\lambda$-uniform, it can be decomposed into $p = |x|/\lambda$ $(s,t)$-path flows of value $\lambda$. These paths are arc-disjoint, and therefore $p$ is minimal. Let us consider a network $\Net^k=(D^k,u^k)$, with $D^k=(V,A^k)$, obtained from $\Net$ and $x$, using only the arcs of colour $k$ and setting $u^k_{a} = x_{a}$, $\forall~a \in A^k$, for $k \in \{1,2\}$; we calculate the maximum flow $x^k$, which is $\lambda$-uniform, and we have the number of path flows of value $\lambda$ on it, given by $p_k = |x^k|/\lambda$. Thus, the number of bichromatic path flows is $p_{12}=p-p_1-p_2$. So, the solution cost is $p_1+p_2+2p_{12}$, which is minimal, since $p_1$ and $p_2$ are maximal.
            \end{shortproof}

            This result can be generalised to bichromatic networks with two distinct flow values on the arcs in which the smaller value divides the greater one, and each flow value is associated to a colour.

            \vspace{10pt}
            \begin{theorem}\label{the:dfcolcstmin2mult}
                \textsc{MinCostCFD} can be solved in polynomial time for bichromatic networks in which every arc of colour $c_i$ has flow value $v_i$, with $i \in \{1,2\}$, $v_1 > v_2$ and $v_2 \mid v_1$.
            \end{theorem}

            \begin{proof}
                Let $\Net$ be a network with the characteristics described, and let $x$ be an $(s,t$)-flow in it. Note that $v_1 = k \cdot v_2$ for some $k \in \mathbb{Z}_+^*$. By replacing each arc of colour $c_1$ with $k$ arcs of colour $c_1$, each with flow value $v_2$ and the same endpoints, we obtain an $(s,t)$-flow $x'$ that is $v_2$-uniform, with $|x'| = |x|$. Each $(s,t)$-path flow of colour $c_1$ in $x$ corresponds to $k$ $(s,t)$-path flows of colour $c_1$ and value $v_2$ in $x'$. Applying the decomposition from \cref{the:mincostcfduni}, we obtain the minimum cost $p_1 + p_2 + 2 \cdot p_{12}$, which corresponds to the minimum cost $p_1/k + p_2 + 2 \cdot p_{12}$ of a decomposition of $x$.
            \end{proof}

            With the above result, we solved a very specific case for networks with two colours ($c_1$ and $c_2$) and two distinct flow values ($v_1$ and $v_2$) in the arcs. The general case for such networks remains an open problem. We may divide it into three subcases as follows:

            \begin{enumerate}
                \item $v_2 \mid v_1$ and there is no correspondence between colours and flow values;
                \item $v_2 \nmid v_1$ and every arc with colour $c_i$ has flow value $v_i$, for $i \in \{1,2\}$;
                \item $v_2 \nmid v_1$ and there is no correspondence between colours and flow values.
            \end{enumerate}

            Considering the first case, even in a DAG, the optimal solution is not greedy, i.e. maximize the monochromatic paths first and then the bichromatic paths with the highest flow value, followed by the monochromatic and bichromatic paths with the lowest flow value. This can be observed in the network in \cref{fig:mincostcfd2v2c}. The thin arcs have flow value $1$, the thick arcs have flow value $2$ and the label on each arc indicates its colour. If we start by extracting the path flow through $scdt$ (monochromatic), we will need four bichromatic path flows to decompose the remaining flow and the cost will be $9$. If we start by extracting the path flows through the $sct$ and $sdt$ (bichromatic), we will need two monochromatic path flows to decompose the remaining flow and the cost will be $6$.

                        \begin{figure}[!ht]
                \centering
                \begin{tikzpicture}[>=stealth, auto=left]
                    \tiny
                    \node [nodeb] (s) [label=180:$s$] at (0.0, 2.0) {};
                    \node [nodeb] (a) [label= 90:$a$] at (1.0, 3.7) {};
                    \node [nodeb] (b) [label= 90:$b$] at (1.3, 3.1) {};
                    \node [nodeb] (c) [label= 90:$c$] at (3.0, 3.0) {};
                    \node [nodeb] (d) [label=270:$d$] at (3.0, 1.0) {};
                    \node [nodeb] (e) [label=270:$e$] at (4.9, 0.9) {};
                    \node [nodeb] (f) [label=270:$f$] at (5.2, 0.3) {};
                    \node [nodeb] (t) [label=  0:$t$] at (6.0, 2.0) {};

                    \path [->, very thick]
                        (s) edge [c02, swap] node{\color{black} 2} (c)
                            edge [c09, swap] node{\color{black} 1} (d)
                        (c) edge [c02]       node{\color{black} 2} (d)
                            edge [c09]       node{\color{black} 1} (t)
                        (d) edge [c02]       node{\color{black} 2} (t);

                    \path [->, c02]
                        (s) edge [bend left=30]        node{\color{black} 2} (a)
                            edge [bend left=15]        node{\color{black} 2} (b)
                        (a) edge [bend left=30]        node{\color{black} 2} (c)
                        (b) edge [bend left=15]        node{\color{black} 2} (c)
                        (d) edge [bend right=15, swap] node{\color{black} 2} (e)
                            edge [bend right=30, swap] node{\color{black} 2} (f)
                        (e) edge [bend right=15, swap] node{\color{black} 2} (t)
                        (f) edge [bend right=30, swap] node{\color{black} 2} (t);
                \end{tikzpicture}
                \caption{Bichromatic network with two distinct flow values in the arcs.}
                \label{fig:mincostcfd2v2c}
            \end{figure}
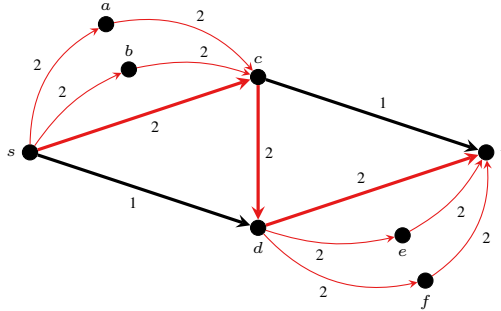

            When we have a network with more than one colour, the solution to the \textsc{MinCostCFD} is not necessarily one that minimizes the number of path flows. We can see this in the \cref{fig:flow3v2c}. Note that its a bichromatic version of the network of \cref{fig:flow3v1c}. Each arc $a$ has a label $(c_a,x_a)$ indicating its colour and flow value. Recall that this flow can be decomposed into $6$ path flows, five of them are bichromatic (three path flows of value $3$ and two path flows of value $1$) and one is monochromatic (with value $2$), making a total cost of $11$. On the other hand, notice that we can extract seven monochromatic path flows, one of value $4$ and six of value $1$, making a total cost of $7$.

                        \begin{figure}[!htb]
                \centering
                \begin{tikzpicture}[>=stealth, auto=left, node distance=1.8cm]
                    \tiny
                    \node [nodeb] (s)  [label=270:$s$]                {};
                    \node [nodeb] (v1) [label=270:$v_1$, right of=s]  {};
                    \node [nodeb] (v2) [label=270:$v_2$, right of=v1] {};
                    \node [nodeb] (v3) [label=270:$v_3$, right of=v2] {};
                    \node [nodeb] (v4) [label=270:$v_4$, right of=v3] {};
                    \node [nodeb] (v5) [label=270:$v_5$, right of=v4] {};
                    \node [nodeb] (v6) [label=270:$v_6$, right of=v5] {};
                    \node [nodeb] (t)  [label=270:$t$,   right of=v6] {};

                    \path [->, very thick]
                        (s)  edge [swap] node{(1,4)} (v1)
                        (v1) edge [swap] node{(1,4)} (v2)
                        (v2) edge        node{(1,4)} (v3)
                        (v3) edge [swap] node{(1,4)} (v4)
                        (v4) edge        node{(1,4)} (v5)
                        (v5) edge [swap] node{(1,4)} (v6)
                        (v6) edge [swap] node{(1,4)} (t);

                    \path [->, thick, c02]
                        (s)  edge [bend left=30]        node{(2,2)} (v1)
                             edge [bend right=40, swap] node{(2,2)} (v3)
                             edge [bend left=50]        node{(2,2)} (v5)
                        (v2) edge [bend left=50]        node{(2,2)} (t)
                        (v4) edge [bend right=40, swap] node{(2,2)} (t)
                        (v6) edge [bend left=30]        node{(2,2)} (t);

                    \path [->, c02]
                        (v1) edge [bend left=25]                node{(2,1)} (v2)
                             edge [bend left=60, looseness=1.3] node{(2,1)} (v2)
                        (v3) edge [bend left=25]                node{(2,1)} (v4)
                             edge [bend left=60, looseness=1.3] node{(2,1)} (v4)
                        (v5) edge [bend left=25]                node{(2,1)} (v6)
                             edge [bend left=60, looseness=1.3] node{(2,1)} (v6);
                \end{tikzpicture}
                \caption{Bichromatic network with 3 distinct flow values on the arcs.}
                \label{fig:flow3v2c}
            \end{figure}
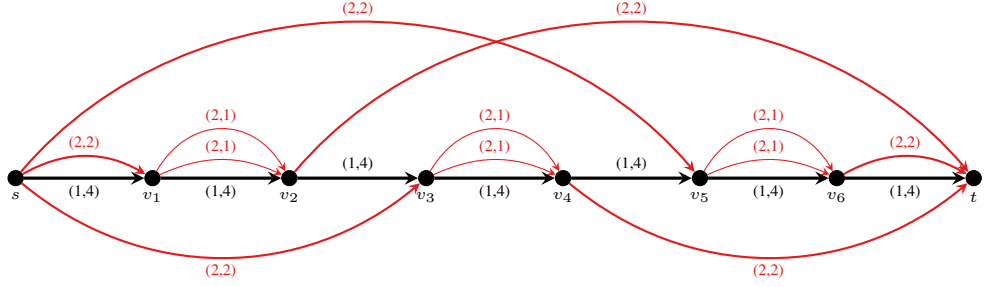

            In networks with two or more colours and three distinct flow values on the arcs, even if they are powers of $2$ (as in the \cref{fig:flow3v1c}), the \textsc{KCostCFD} problem is $\NP$-Complete. We show a reduction to this problem from the \textsc{KSplittableFlow} problem, which is $\NP$-Complete under the same conditions.

            \vspace{10pt}
            \begin{theorem}\label{the:dfcolcstk2c3v}
                \textsc{KCostCFD} is $\NP$-Complete if the network has at least two colours and the flow value on each arc is either $1$, $2$ or $4$.
            \end{theorem}

            \begin{shortproof}
                We will show a polynomial reduction from the \textsc{KSplittableFlow} to the \textsc{KCostCFD} problem. Given an instance $\langle \Net=(D,u),x,k \rangle$, with $D=(V,A)$ with two special vertices $s,t \in V$, of \textsc{KSplittableFlow}, where the flow values on each arc are in $\{1,2,4\}$, we construct an instance $\langle \Net'=(D',u'),x',k' \rangle$ of \textsc{KCostCFD} with $q \geq 2$ colours as follows:

                \begin{itemize}
                    \item $D' = (V', A', c)$, where:\\
                          $V' = V \cup \{z_i \mid 2 \leq i \leq q\}$ and $A' = A \cup \{sz_i, z_it \mid 1 \leq i \leq q\}$\\
                          and the colouring $c: A' \rightarrow \{1, \ldots, q\}$, where $c(a) = 1$, if $a \in A$; or $c(a) = i$, if $z_i$ is one of the endpoints of $a$;
                    \item we define $x'(a) = x(a)$, if $a \in A$; or $x'(a) = 1$, otherwise;
                    \item we define $u'(a) = x'(a)$, for all $a \in A'$; \item we set $k'=k+q-1$.
                \end{itemize}

                The network $\Net'$ described above has $n+q-1$ vertices and $m+2q-2$ arcs, where $n = |V|$ and $m = |A|$. So, it is polynomial in the input size. This construction is illustrated in \cref{fig:redksplitkcostcfd} (the labels next to the arcs indicate their colours).

                            \begin{figure}[!htb]
                \centering
                \begin{tikzpicture}[>=stealth, auto=left, x=1.2cm, y=1.2cm]
                    \tiny

                    \draw [draw=c08, fill=none,   dashed, rounded corners] (0.3, 0.8) rectangle (6.7, 4.8);
                    \draw [draw=c08, fill=c08!10, dashed, rounded corners] (0.5, 1.0) rectangle (6.5, 3.0);

                    \node[nodeb] (s)  [label=180:$s$]   at (1.0, 2.0) {};
                    \node[nodeb] (t)  [label=  0:$t$]   at (6.0, 2.0) {};
                    \node[nodeb] (z2) [label=300:$z_2$] at (3.5, 3.5) {};
                    \node[nodeb] (zq) [label=300:$z_q$] at (3.5, 4.5) {};

                    \node (s1) at (2.5, 2.5) {};
                    \node (s2) at (2.5, 1.5) {};
                    \node (t1) at (4.5, 2.5) {};
                    \node (t2) at (4.5, 1.5) {};
                    \node (sr) at (2.0, 2.1) {\vdots};
                    \node (tr) at (5.0, 2.1) {\vdots};
                    \node (zr) at (3.5, 4.1) {\vdots};
                    \node (l1) at (6.2, 1.3) {$\Net$};
                    \node (l2) at (6.4, 4.5) {$\Net'$};

                    \path [->]
                        (s)  edge [c09, bend left=20,  pos=0.8]       node{\color{black} 1} (s1)
                             edge [c09, bend right=20, pos=0.8, swap] node{\color{black} 1} (s2)
                             edge [c02, bend left=30]                 node{\color{black} 2} (z2)
                             edge [c06, bend left=38]                 node{\color{black} q} (zq)
                        (z2) edge [c02, bend left=30]                 node{\color{black} 2} (t)
                        (zq) edge [c06, bend left=38]                 node{\color{black} q} (t)
                        (t1) edge [c09, bend left=20,  pos=0.2]       node{\color{black} 1} (t)
                        (t2) edge [c09, bend right=20, pos=0.2, swap] node{\color{black} 1} (t);
                \end{tikzpicture}
                \caption{Reduction from \textsc{KSplittableFlow} to the \textsc{KCostCFD}.}
                \label{fig:redksplitkcostcfd}
            \end{figure}
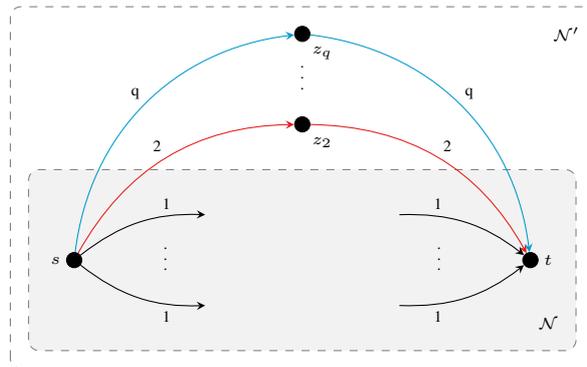

                The $(s,t)$-flow $x$ in $\Net$ can be decomposed into $k$ $(s,t)$-flow paths if and only if $x'$ in $\Net'$ can be decomposed into $(s,t)$-flow paths with a cost $k'=k+q-1$.

                Let $x^1, \ldots, x^k$ be a decomposition of $x$ into $(s,t)$-flow paths. Every flow path in $\Net$ is a flow path of cost $1$ in $\Net'$. The paths $sz_it$, for $2 \leq i \leq q$, are monochromatic and therefore the flow along each of them has cost of $1$. Thus, $x'$ admits a decomposition into $q$ flow paths, $x^1, \ldots, x^k, \ldots, x^{q-1}$, which are monochromatic, and the total cost of this decomposition is $k + q - 1$.

                Now consider a decomposition of $x'$ into $r$ $(s,t)$-path flows with cost $k+q-1$. Observe that the paths $sz_it$ in $\Net'$, for $2 \leq i \leq q$, are arc-disjoint and monochromatic. Thus, they contribute with a cost of $q-1$ to the total cost of the solution. The other flow paths, which are also monochromatic, have a total cost of $k=r-(q-1)$. Since these are also path flows in $\Net$, they correspond to a decomposition of $x$ into $k$ path flows.
            \end{shortproof}

        \subsection{Networks with at least three colours}

            For a network with three colours and a $\lambda$-uniform $(s,t)$-flow, the \textsc{MinCostCFD} problem remains $\NP$-hard. We show a reduction to the \textsc{KCostCFD} problem from the \textsc{Weak-2-Linkage} problem (described below), which, according to \citeonline{FortuneKLinkage}, is $\NP$-complete.

            \problem{Weak 2-Linkage}
                    {A digraph $D=(V,A)$ with $4$ special vertices $u_1$, $u_2$, $v_1$ and $v_2 \in V$.}
                    {Does $D$ contain a pair of arc disjoint paths $P_1$ and $P_2$, such that $P_i$ is a $(u_i,v_i)$-path for $i \in \{1,2\}$?}

            \vspace{-6pt}

            \begin{theorem}\label{the:kcostdfc3c1v}
                The problem \textsc{KCostCFD} is $\NP$ complete even when the flow is uniform and there are only three distinct colors on the edges of the network.
            \end{theorem}

            \begin{shortproof}
                We will show a polynomial reduction from the \textsc{Weak 2-Linkage} problem to \textsc{KCostCFD}. Given an instance $\langle D=(V,A), u_1, u_2, v_1, v_2 \rangle$, where $u_1, u_2, v_1, v_2 \in V$, of the \textsc{Weak 2-Linkage} problem, we will construct a network $\Net=(D', u \equiv \lambda)$ with $3$ colours, and define a $\lambda$-uniform $(s,t)$-flow $x$ in it, that is, an instance of \textsc{KCostCFD}. Let $n = |V|$ and $m = |A|$. Note that $n \geq 4$ and $m \geq 2$. Otherwise, we would have a instance ``No'' of the \textsc{Weak 2-Linkage} problem. Initially, we will construct the multidigraph $D'=(V',A',c)$ with a colouring $c: A' \rightarrow \{1,2,3\}$ as follows:

                \begin{itemize}
                    \item $V' = V \cup \{s, t, s_1, s_2, t_1, t_2\}$;
                    \item The arcs in $A'$ will be added in the following steps:
                          \begin{itemize}
                              \item Every arc in $A$ also becomes an arc in $A'$ with colour $3$;
                              \item Add $4$ arcs: from $s$ to $u_1$, from $s$ to $u_2$, from $v_1$ to $t$, and from $v_2$ to $t$. Each of these arcs, with endpoints $u_i$ or $v_i$, has colour $i$, for $i \in \{1,2\}$;
                              \item Add $m-2$ arcs with colour $1$ from $s$ to $s_1$; $m-2$ arcs with colour $3$ from $s_1$ to $s_2$ and from $t_1$ to $t_2$; and $m-2$ arcs with colour $2$ from $t_2$ to $t$;
                              \item Add $d^+_D(u_i)-1$ arcs with colour $2$ from $s_2$ to $u_i$, and $d^-_D(v_i)-1$ arcs with colour $1$ from $v_i$ to $t_1$, for $i \in \{1,2\}$;
                              \item For every $v \in V \setminus \{u_1, u_2\}$, add $d^+_D(v)$ arcs of colour $2$ from $s_2$ to $v$;
                              \item For every $v \in V \setminus \{v_1, v_2\}$, add $d^-_D(v)$ arcs of colour $1$ from $v$ to $t_1$;
                          \end{itemize}
                \end{itemize}

                To complete the definition of $\Net$, we set $u(a) = x(a) = \lambda$ for all $a \in A'$. Note that $b_x(v) = 0$ for all $v \in V' \setminus \{s,t\}$, and that $b_x(s) = -b_x(t) = m \cdot \lambda$, which is the value of the $(s,t)$-flow $x$. This can be decomposed into $m$ path flows, each with value $\lambda$. The resulting multidigraph $D'$ has $n + 6$ vertices and $6(m-2) + m + 4 = 7m - 8$ arcs. The construction, illustrated in \cref{fig:red2linkagekcostcfd}, is therefore polynomial in the size of the input.

                Observe that for each arc $uv$ in $A$, there is a three-colour path from $s$ to $u$ or from $v$ to $t$ in $D'$. Thus, the cost of any decomposition of the $(s,t)$-flow $x$ into $m$ $(s,t)$-path flows is at most $3m$. It will show that $D$ has two arc-disjoint paths, one from $u_1$ to $v_1$ and another from $u_2$ to $v_2$, if and only if the $(s,t)$-flow $x$ admits a decomposition into $m$ $(s,t)$-path flows with cost $3m-2$.

                Assume that $D$ has the two arc-disjoint paths. These are also paths in $D'$. By adding the edges from $s$ to $u_i$, and from $v_i$ to $t$, both with colour $i$ for $i \in \{1,2\}$, we obtain two bichromatic path flows. The other $m-2$ path flows have three colours. Therefore, the cost of this decomposition is $3(m-2) + 2 \cdot 2 = 3m-2$.

                Now, assume that the $(s,t)$-flow $x$ can be decomposed into $m$ path flows, $x^1, \ldots, x^m$, with a cost of $3m-2$. Since $x$ is $\lambda$-uniform, $m$ is minimal and the path flows in the decomposition are pairwise arc-disjoint. Each path flow has between two and three colours. Since there can be at most two path flows with two colours, there are at least $m-2$ path flows with three colours. In fact, there are exactly $m-2$ path flows with $3$ colours. Otherwise, their cost would be $3(m-2+k) + 2(2-k) = 3m - 2 + k > 3m - 2$, for $k \in \{1,2\}$. The arcs of colour $3$ are either in $A$, or from $s_1$ to $s_2$ or from $t_1$ to $t_2$. These last two types must necessarily be in path flows with three colours. The path flows with two colours have colour $3$ and a second colour (either $1$ or $2$). In these, the only arcs with a colour different from $3$ are the first one, from $s$ to $u_i$, and the last one, from $v_i$ to $t$, for $i \in \{1,2\}$. By removing these initial and final arcs, we obtain two paths, from $u_1$ to $v_1$ and from $u_2$ to $v_2$, that are arc-disjoint in $D$.
            \end{shortproof}

                            \begin{figure}[!htb]
                    \centering
                    \begin{subfigure}[b]{0.35\textwidth}
                        \centering
                        \begin{tikzpicture}[>=stealth, auto=left]
                            \tiny
                            \node [nodeb] (u1) [label=120:$u_1$] at (2.0, 4.0) {};
                            \node [nodeb] (u2) [label=250:$u_2$] at (2.0, 2.0) {};
                            \node [nodeb] (z)  [label=270:$z$]   at (4.0, 3.0) {};
                            \node [nodeb] (v1) [label= 70:$v_1$] at (6.0, 4.0) {};
                            \node [nodeb] (v2) [label=290:$v_2$] at (6.0, 2.0) {};

                            \path [->]
                                (u1) edge (z)
                                     edge (v1)
                                (u2) edge (z)
                                     edge (u1)
                                (z)  edge (v1)
                                (v1) edge (v2)
                                (v2) edge (z);
                        \end{tikzpicture}
                        \caption*{$D$}
                    \end{subfigure}
                    \begin{subfigure}[b]{0.64\textwidth}
                        \centering
                        \begin{tikzpicture}[>=stealth, auto=left]
                            \tiny
                            \node [nodeb] (s)  [label=180:$s$]   at (0.0, 3.0) {};
                            \node [nodeb] (s1) [label= 90:$s_1$] at (0.0, 5.0) {};
                            \node [nodeb] (s2) [label= 90:$s_2$] at (2.0, 6.0) {};
                            \node [nodeb] (t)  [label=  0:$t$]   at (8.0, 3.0) {};
                            \node [nodeb] (t1) [label= 90:$t_1$] at (6.0, 6.0) {};
                            \node [nodeb] (t2) [label= 90:$t_2$] at (8.0, 5.0) {};
                            \node [nodeb] (u1) [label=100:$u_1$] at (2.0, 4.0) {};
                            \node [nodeb] (u2) [label=260:$u_2$] at (2.0, 2.0) {};
                            \node [nodeb] (z)  [label=270:$z$]   at (4.0, 3.0) {};
                            \node [nodeb] (v1) [label= 80:$v_1$] at (6.0, 4.0) {};
                            \node [nodeb] (v2) [label=280:$v_2$] at (6.0, 2.0) {};

                            \path[c08, opacity=0.6]
                                (s)  edge [decorate, bend left=10]  (u1)
                                (u1) edge [decorate]                (v1)
                                (v1) edge [decorate, bend left=10]  (t);

                            \path[c08, decoration={zigzag, amplitude=.8mm, segment length=1.5mm}, opacity=0.6]
                                (s)  edge [bend right=10, decorate] (u2)
                                (u2) edge [decorate]                (z)
                                (z)  edge [decorate]                (v1)
                                (v1) edge [decorate]                (v2)
                                (v2) edge [bend right=10, decorate] (t);

                            \path [->]
                                (s1) edge [bend left=20]          node{\color{black} 3} (s2)
                                     edge [bend left=10]                                (s2)
                                     edge                                               (s2)
                                     edge [bend right=10]                               (s2)
                                     edge [bend right=20]                               (s2)
                                (t1) edge [bend left=20]          node{\color{black} 3} (t2)
                                     edge [bend left=10]                                (t2)
                                     edge                                               (t2)
                                     edge [bend right=10]                               (t2)
                                     edge [bend right=20]                               (t2)
                                (u1) edge [pos=0.7, swap]         node{\color{black} 3} (z)
                                     edge [pos=0.2, swap]         node{\color{black} 3} (v1)
                                (u2) edge [swap]                  node{\color{black} 3} (z)
                                     edge [swap]                  node{\color{black} 3} (u1)
                                (z)  edge [pos=0.7, swap]         node{\color{black} 3} (v1)
                                (v1) edge                         node{\color{black} 3} (v2)
                                (v2) edge                         node{\color{black} 3} (z);

                            \path [->, c06]
                                (s)  edge [bend left=10]          node{\color{black} 1} (u1)
                                     edge [bend left=20]          node{\color{black} 1} (s1)
                                     edge [bend left=10]                                (s1)
                                     edge                                               (s1)
                                     edge [bend right=10]                               (s1)
                                     edge [bend right=20]                               (s1)
                                (v1) edge [bend left=10]          node{\color{black} 1} (t)
                                     edge                                               (t1)
                                (u1) edge [bend left=10, pos=0.9] node{\color{black} 1} (t1)
                                (z)  edge [bend left=10]                                (t1)
                                     edge                                               (t1)
                                     edge [bend right=10]                               (t1);

                            \path [->, c02]
                                (s)  edge [bend right=10, swap]   node{\color{black} 2} (u2)
                                (v2) edge [bend right=10, swap]   node{\color{black} 2} (t)
                                (t2) edge [bend left=20]          node{\color{black} 2} (t)
                                     edge [bend left=10]                                (t)
                                     edge                                               (t)
                                     edge [bend right=10]                               (t)
                                     edge [bend right=20]                               (t)
                                (s2) edge                                               (u1)
                                     edge [bend right=30]                               (u2)
                                     edge [bend left=10, pos=0.1] node{\color{black} 2} (v1)
                                     edge [bend left=10]                                (v2)
                                     edge                                               (z);
                        \end{tikzpicture}
                        \caption*{$\Net$}
                    \end{subfigure}
                    \caption{Exemple of reduction from \textsc{Weak 2-Linkage} to \textsc{KCostCFD}.}
                    \label{fig:red2linkagekcostcfd}
                \end{figure}

            \cref{fig:red2linkagekcostcfd} provides an example of the reduction presented in \cref{the:kcostdfc3c1v}. On the left we have an instance of the \textsc{Weak 2-Linkage} problem; and on the right an instance of the \textsc{KCostCFD} problem, where all arcs have capacity and flow equal to $\lambda$. The labels on the arcs indicate their colours. The flow can be decomposed into $m = 7$ arc-disjoint path flows (with $2$ bichromatic and $5$ trichromatic path flows), with cost $19$.

            The result of \cref{the:kcostdfc3c1v} can be generalised to the case where there are multiple flow values in the set of arcs of the network. To do this, it suffices to add an arc with colour $3$ from $s$ to $t$ for each new flow value. Thus, we have:

            \vspace{10pt}
            \begin{corollary}\label{cor:kcostcfd3c}
                \textsc{KCostCFD} is $\NP$-Complete for networks with three colours.
            \end{corollary}

            According to \cite{BangJensenSubdivision}, the \textsc{Weak 2-Linkage} problem is $\NP$-Complete even in digraphs with maximum degree $3$. The network from the reduction proposed in \cref{the:kcostdfc3c1v} can be modified such that, except for $s$ and $t$, every vertex has maximum degree $6$. This can be achieved by taking an instance of the \textsc{Weak 2-Linkage} problem where every vertex has a degree of at most $3$, and transforming each multipath $ss_1s_2$ and $t_1t_2t$ into $m-2$ multipaths of the form $ss_1^is_2^i$ and $m-2$ multipaths of the form $t_1^it_2^it$, for $1 \leq i \leq m-2$. Except for $s$ and $t$, each vertex in these paths has degree $2$ and each vertex from the original digraph will have its degree doubled. Therefore, we have the following result:

            \vspace{10pt}
            \begin{corollary}\label{cor:kcostcfdg6}
                The \textsc{KCostCFD} problem is $\NP$-Complete even when restricted to networks with $3$ colours where, except for $s$ and $t$, every vertex has degree at most $6$.
            \end{corollary}

            In \citeonline{FortuneKLinkage}, the authors showed that \textsc{Weak 2-Linkage} problem can be solved in polynomial time in $DAGs$. We also decided to analyse the \textsc{KCostCFD} problem for $DAGs$ with $3$ colours. Let us consider the network illustrated in \cref{fig:acyclicnetwork3c1v}. In this network, the label on each arc represents its colour. Note that each $(s,t)$-path flow contains at least two colours. Depending on how the first path is extracted, the remaining path may have $2$ or $3$ colours. To the optimal solution, we must extract the path flows through the paths $sacdfgt$ and $sbcefht$, both bichromatic. By selecting part of this solution, and considering the segments of the path flows starting from vertex $c$, notice that an optimal decomposition for the highlighted area is not achieved, as there are two bichromatic paths, while the optimal solution for this segment consists of one monochromatic path and one bichromatic path.

                        \begin{figure}[!htb]
                \centering
                \begin{tikzpicture}[>=stealth, auto=left, x=1.5cm, y=1.0cm]
                    \tiny
                    \draw [draw=c08, fill=c08!10, rounded corners, dashed] (1.8, 2.7) rectangle (6.2, -0.7);
                    \node [nodeb] (s) [label=180:$s$] at (0.0, 1.0) {};
                    \node [nodeb] (a) [label= 90:$a$] at (1.0, 2.0) {};
                    \node [nodeb] (b) [label=270:$b$] at (1.0, 0.0) {};
                    \node [nodeb] (c) [label= 90:$c$] at (2.0, 1.0) {};
                    \node [nodeb] (d) [label= 90:$d$] at (3.0, 2.0) {};
                    \node [nodeb] (e) [label=270:$e$] at (3.0, 0.0) {};
                    \node [nodeb] (f) [label= 90:$f$] at (4.0, 1.0) {};
                    \node [nodeb] (g) [label= 90:$g$] at (5.0, 2.0) {};
                    \node [nodeb] (h) [label=270:$h$] at (5.0, 0.0) {};
                    \node [nodeb] (t) [label= 90:$t$] at (6.0, 1.0) {};

                    \path [c08, decoration={snake, amplitude=.8mm, segment length=2.0mm}, opacity=0.6]
                        (s) edge [decorate, bend left=15]  (a)
                        (a) edge [decorate, bend left=15]  (c)
                        (c) edge [decorate, bend left=15]  (d)
                        (d) edge [decorate, bend left=15]  (f)
                        (f) edge [decorate, bend left=15]  (g)
                        (g) edge [decorate, bend left=15]  (t);

                    \path [c08, decoration={zigzag, amplitude=.8mm, segment length=1.5mm}, opacity=0.6]
                        (s) edge [decorate, bend right=15] (b)
                        (b) edge [decorate, bend right=15] (c)
                        (c) edge [decorate, bend right=15] (e)
                        (e) edge [decorate, bend right=15] (f)
                        (f) edge [decorate, bend right=15] (h)
                        (h) edge [decorate, bend right=15] (t);

                    \path [->]
                        (s) edge [bend left=15]             node{\color{black} 3} (a)
                            edge [bend right=15, c06, swap] node{\color{black} 1} (b)
                        (a) edge [bend left=15,  c02]       node{\color{black} 2} (c)
                        (b) edge [bend right=15, c02, swap] node{\color{black} 2} (c)
                        (c) edge [bend left=15]             node{\color{black} 3} (d)
                            edge [bend right=15, c02, swap] node{\color{black} 2} (e)
                        (d) edge [bend left=15]             node{\color{black} 3} (f)
                        (e) edge [bend right=15, c02, swap] node{\color{black} 2} (f)
                        (f) edge [bend left=15,  c02]       node{\color{black} 2} (g)
                            edge [bend right=15, c06, swap] node{\color{black} 1} (h)
                        (g) edge [bend left=15,  c02]       node{\color{black} 2} (t)
                        (h) edge [bend right=15, c06, swap] node{\color{black} 1} (t);
                \end{tikzpicture}
                \caption{Example of network over a \emph{DAG} with 3 colours and a uniform flow.}
                \label{fig:acyclicnetwork3c1v}
            \end{figure}
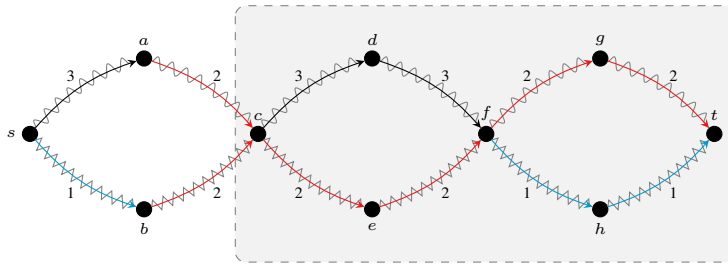

            As observed in the previous example, a $\lambda$-uniform flow in an acyclic network with $3$ distinct colours on its arcs, the problem does not exhibit the \emph{optimal substructure}, in which part of the optimal solution is itself an optimal solution for the respective subproblem, that is, an optimal solution to the problem contains within it optimal solutions to subproblems. This property is characteristic of problems that can be efficiently solved by greedy algorithms or dynamic programming technics (see \cite{Cormen}).

            We show that \emph{KCostCFD} é $\NP$-Complete, even for $\lambda$-uniform flows, in acyclic networks with at least $5$ colours. That is done with a reduction from the \textsc{1-in-3SAT} problem (described bellow). According to \cite{Schaefer}, this problem is $\NP$-Complete.

            \problem{1-in-3SAT}
                    {A 3CNF formula $\varphi$ with $m$ clauses and $n$ variables.}
                    {Does there exist a truth assignment to the variables of $\varphi$ such that just one literal per clause is true?}

            \vspace{-12pt}

            \begin{theorem}\label{the:kcostcfddaguni}
                \textsc{KCostCFD} problem is $\NP$-Complete, even for acyclic networks with $5$ or more colours and $\lambda$-uniform $(s,t)$-flows.
            \end{theorem}

            \begin{shortproof}
                We will show a polynomial reduction from the \textsc{1-in-3SAT} to \textsc{KCostCFD}. Given an instance $\varphi$ of the first, a $3$CNF formula with $m$ clauses and $n$ variables, we will construct a network $\Net=(D', u \equiv \lambda)$, with $n + 2$ colours and a $\lambda$-uniform $(s,t)$-flow $x$.

                Initially, we create two special vertices $s$ and $t$, and a vertex $v_i$ for each variable $v_i$, for $1 \leq i \leq n$, with two arcs from $s$ to it, one with colour $1$ and another with colour $2$. For each clause $C_j$ in $\varphi$, for $1 \leq j \leq m$, we create a gadget with vertices $c_j^s$, $c_j^t$, and three arcs from the first to the second, one with colour $1$ and two with colour $2$.

                Now we need to construct the paths from $s$ to $t$. For each variable $v_i$, we will construct $2$ such paths. If the literal $v_i$ appears in any clause, let $C_j, \ldots, C_k$ be the sequence of clauses that contain the literal $v_i$, for $1 \leq j \leq k \leq m$. Add an arc from vertex $v_i$ to $c_j^s$. If $j < k$, then add an arc from $c_x^t$ to $c_{x+1}^s$ for $j \leq x < k$. Finally, add an arc from $c_k^t$ to $t$. If there are no clauses with the literal $v_i$, create an arc from $v_i$ to $t$. Proceed similarly for the literal $\nao{v_i}$. All arcs mentioned here have colour $i+2$. Since the instance of \textsc{1-in-3SAT} has at least $3$ variables, the network here defined has at least $5$ colours.

                By the construction here described, illustrated in \cref{fig:red13satkcostcfd}, in addition to $s$ and $t$, $n$ vertices are created, one for each variable, and $2m$ vertices, two for each clause. For each variable, $4$ arcs are created, and for each clause, $6$ arcs are created. Thus, the resulting network has exactly $n + 2m + 2$ vertices and $4n + 6m$ arcs. Therefore, the construction is polynomial in the input size.

                Finally, we define $u(a) = x(a) = \lambda$ for every arc $a$ in the network. Thus, the $(s,t)$-flow $x$ is $\lambda$-uniform and has value $2n\lambda$. It can be decomposed into $2n$ $(s,t)$-path flows with value $\lambda$. Since the path flows are arc-disjoint, the number $2n$ is minimal. Each path has at least two colours (either $1$ or $2$ from $s$ to a vertex $v_i$, and $i+2$ leaving $v_i$). Thus, the cost of any decomposition of $x$ into path flows is at least $4n$. It will be shown that the formula $\varphi$ in the \textsc{1-in-3SAT} problem is satisfiable if and only if $x$ can be decomposed into flow paths with a cost $4n$.

                Observe that the number of path flows with value $\lambda$ through the gadget corresponding to a clause in $\varphi$ is exactly $3$. One of these paths goes through the arc with colour $1$, and the other two paths go through an arc with colour $2$.

                Let us consider an assignment of values to the variables of $\varphi$ such that only one literal per clause is true. For each variable $v_i$ of $\varphi$, two $(s,t)$-path flows must be taken, both with colours $i+2$ and a second colour ($1$ for one path and $2$ for the other). For each variable $v_i$, at least one of its literals ($v_i$ or $\nao{v_i}$) must be present in some clause of $\varphi$. Consequently, at least one of the $(s,t)$-path flows through $v_i$ must go through at least one of the clause gadgets. Now, will describe how each one of these path flows must be taken.

                If $v_i$ is \emph{false} (resp. \emph{true}), the first $(s,t)$-flow path, with colours $1$ and $i+2$, must go through the gadgets corresponding to the clauses that contain the literal $\nao{v_i}$ (resp. $v_i$), if there is some. Otherwise, it should go through the arc from $v_i$ to $t$. The second $(s,t)$-flow path, with colours $2$ and $i+2$, must go through the gadgets corresponding to the clauses that contain the literal $v_i$ (resp. $\nao{v_i}$), if there is some. Otherwise, it should go through the arc from $v_i$ to $t$.

                Now assume the $(s,t)$-flow $x$ admits a decomposition into $(s,t)$-path flows with cost equal to $4n$. Since the minimum number of flow paths in a decomposition of $x$ is $2n$ and the cost of each one of them is at least $2$, there are exactly bichromatic $2n$ $(s,t)$-path flows. Notice that each $(s,t)$-path flow starts with two colours (either $1$ or $2$) from $s$ to a vertex $v_i$, and colour $i+2$ from there, with $1 \leq i \leq n$. Thus, if an $(s,t)$-path flow has colours $1$ and $2$, then it has at least three colours. For each vertex $v_i$, there is at least one and at most two arcs from it to a gadget of clause. Select (arbitrarily, if there is more than one option) an $(s,t)$-path flow that has an arc from $v_i$ to $c_j^s$, with $1 \leq j \leq m$. If clause $C_j$ contains the literal $v_i$ and the selected $(s,t)$-flow path has colour $1$ (resp. $2$), assign \emph{true} (resp. \emph{false}) to the variable $v_i$ in $\varphi$. Similarly, if clause $C_j$ contains the literal $\nao{v_i}$ and the selected $(s,t)$-path flow has colour $1$ (resp. $2$), assign \emph{false} (resp. \emph{true}) to the variable $v_i$ in $\varphi$. Since each gadget of clause has only one edge of colour $1$, only one literal from the corresponding clause in $\varphi$ will have the value \emph{true}, which is a condition for the formula $\varphi$ to be satisfied in the \textsc{1-in-3SAT} problem.
            \end{shortproof}

            Figure \cref{fig:red13satkcostcfd} brings an example of the reduction from \textsc{1-in-3SAT} to \textsc{KCostCFD}, with the network obtained from a formula $\varphi$. In this figure, the label on each arc indicates its colour. The highlighted $(s,t)$-path flows are those that pass through each vertex $v_i$, with $1 \leq i \leq n$, and then through a clause vertex (in the example, the vertex $c_1^s$). From these paths, based on the reduction from \cref{the:kcostcfddaguni}, we get a truth assignment that makes $\varphi$ satisfiable in \textsc{1-in-3SAT} problem. Observe that the paths through vertices $v_1$ and $v_2$ have colour $2$, and the literals $v_1$ and $v_2$ are in $C_1$. So, the value \emph{false} should be assigned to these two variables. The path that passes through $v_3$ has colour $1$, and the literal $v_3$ is in $C_1$. Thus, the value \emph{true} should be assigned to the variable $v_3$.

                        \begin{figure}[!htb]
                \centering
                \begin{tikzpicture}[>=stealth, auto=left, x=1.3cm, y=1.3cm]
                    \tiny
                    \def\clausula#1#2#3{
                        \node[nodeb] (c#1s) [label=90:$c_#1^s$] at (#2, #3)     { };
                        \node[nodeb] (c#1t) [label=90:$c_#1^t$] at (#2 + 2, #3) { };

                        \path[->]
                            (c#1s) edge [c06, bend left=40]  node{\color{black} 1} (c#1t)
                                   edge [c02]                node{\color{black} 2} (c#1t)
                                   edge [c02, bend right=40] node{\color{black} 2} (c#1t);
                    }

                    \node        (f) at (4.5, 6.5) {$\varphi = \underbrace{(v_1 \vee v_2 \vee v_3)}_{C_1} \wedge
                                                               \underbrace{(v_1 \vee \nao{v_2} \vee \nao{v_3})}_{C_2}$};

                    \node[nodeb] (s)  [label=180:$s$]  at (0.0, 3.0) { };
                    \node[nodeb] (t)  [label=  0:$t$]  at (9.0, 3.0) { };
                    \node[nodeb] (v1) [label=90:$v_1$] at (2.2, 4.5) { };
                    \node[nodeb] (v2) [label=90:$v_2$] at (2.5, 3.0) { };
                    \node[nodeb] (v3) [label=90:$v_3$] at (2.2, 1.5) { };

                    \clausula{1}{4.5}{4.0};
                    \clausula{2}{4.5}{2.0};

                    \draw [dashed, c08] (1.0, 6.0) -- (1.0, 1.0);
                    \draw [dashed, c08] (3.5, 6.0) -- (3.5, 1.0);
                    \draw [dashed, c08] (7.5, 6.0) -- (7.5, 1.0);

                    \node  [text=c08] (lbv) at (2.2, 5.8) {Variables};
                    \node  [text=c08] (lbc) at (5.5, 5.8) {Clauses};

                    \path[c08, decoration={coil, amplitude=.8mm, segment length=1mm}, opacity=0.5]
                        (s)   edge [decorate, bend right=15]   (v1)
                        (v1)  edge [decorate, out=0, in=140]   (c1s)
                        (c1s) edge [decorate, bend right=40]   (c1t)
                        (c1t) edge [decorate, out=270, in=140] (c2s)
                        (c2s) edge [decorate]                  (c2t)
                        (c2t) edge [decorate]                  (t);

                    \path[c08, decoration={zigzag, amplitude=.8mm, segment length=1mm}, opacity=0.5]
                        (s)   edge [decorate, bend right=15] (v2)
                        (v2)  edge [decorate, out=0, in=180] (c1s)
                        (c1s) edge [decorate]                (c1t)
                        (c1t) edge [decorate, bend left=20]  (t);

                    \path[c08, decoration={snake, amplitude=.8mm, segment length=1.5mm}, opacity=0.5]
                        (s)   edge [decorate, bend left=15]   (v3)
                        (v3)  edge [decorate, out=20, in=220] (c1s)
                        (c1s) edge [decorate, bend left=40]   (c1t)
                        (c1t) edge [decorate, bend right=10]  (t);

                    \path[->, c06, thick]
                        (s) edge [bend left=15, pos=0.7]          node{\color{black} 1} (v1)
                            edge [bend left=15, pos=0.7]          node{\color{black} 1} (v2)
                            edge [bend left=15, pos=0.7]          node{\color{black} 1} (v3);

                    \path[->, c02]
                        (s) edge [bend right=15, pos=0.7, swap]   node{\color{black} 2} (v1)
                            edge [bend right=15, pos=0.7, swap]   node{\color{black} 2} (v2)
                            edge [bend right=15, pos=0.7, swap]   node{\color{black} 2} (v3);

                    \path [->, c03]
                        (v1)  edge [out=0, in=140, pos=0.3, swap] node{\color{black} 3} (c1s)
                              edge [bend left=50, pos=0.1]        node{\color{black} 3} (t)
                        (c1t) edge [out=270, in=140, swap]        node{\color{black} 3} (c2s)
                        (c2t) edge                                node{\color{black} 3} (t);

                    \path [->, c10]
                        (v2)  edge [out=0, in=180, pos=0.3]       node{\color{black} 4} (c1s)
                              edge [out=0, in=180, pos=0.3, swap] node{\color{black} 4} (c2s)
                        (c1t) edge [bend left=20, pos=0.2]        node{\color{black} 4} (t)
                        (c2t) edge [bend left=25, pos=0.2]        node{\color{black} 4} (t);

                    \path [->, c11]
                        (v3)  edge [out=20, in=220, pos=0.2]      node{\color{black} 5} (c1s)
                              edge [out=0, in=220, pos=0.3, swap] node{\color{black} 5} (c2s)
                        (c1t) edge [bend right=10, pos=0.3, swap] node{\color{black} 5} (t)
                        (c2t) edge [bend right=25, pos=0.2, swap] node{\color{black} 5} (t);
                \end{tikzpicture}
                \caption{Reduction from \textsc{1-em-3SAT} to \textsc{KCostCFD}.}
                \label{fig:red13satkcostcfd}
            \end{figure}
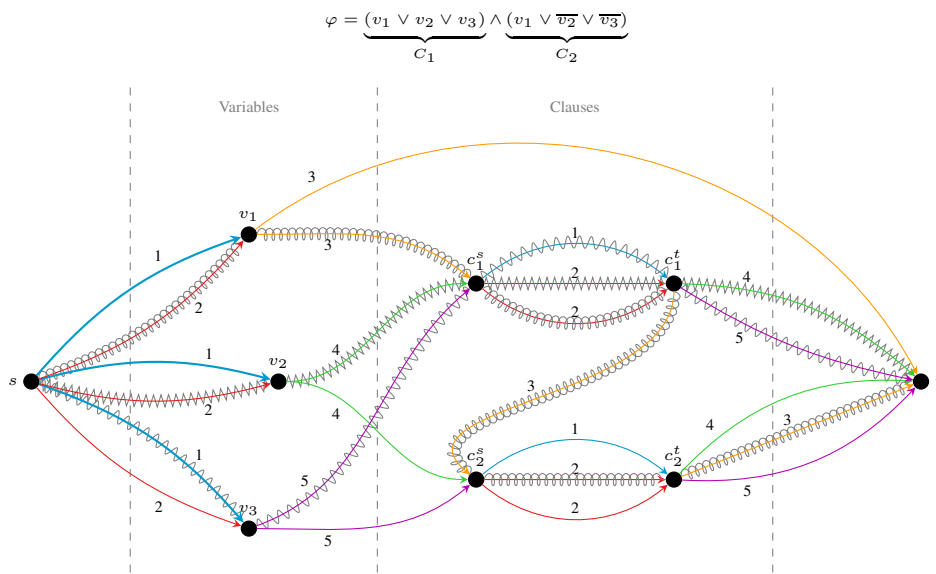

            From the reduction proposed in \cref{the:kcostcfddaguni}, notice that every vertex other than $s$ and $t$ have degree $4$ or $6$. Thus, we have the following:

            \vspace{10pt}
            \begin{corollary}\label{cor:kcostcfd5cg46}
                \textsc{KCostCFD} is $\NP$-Complete, even when restricted to acyclic networks with $5$ or more colours in which, except for $s$ and $t$, every vertex has degree $4$ or $6$.
            \end{corollary}

            For a $\lambda$-uniform $(s,t)$-flow $x$ in an acyclic arc-coloured network, if the colouring function $c$ is injective, the minimum cost of the decomposition is $m$. This is because the decomposition that minimizes the cost consists of $|x|/\lambda$ arc-disjoint path flows. In this case, each arc appears in exactly one path flow of the decomposition.

            The \textsc{KCostCFD} problem remains open for $\lambda$-uniform flows in acyclic networks with $3$ or $4$ colours. \cref{tbl:results} provides a summary of the results discussed for the \textsc{KCostCFD} problem in this work, relating the number of colours to the number of distinct flow values on the arcs of the network. Except for \cref{lem:ksplitmult} and \cref{the:ksplit2v,the:ksplit3v1c} the other results are ours.

                    \begin{table}[!htb]
            \centering
            \begin{tikzpicture}[x=1.5cm, y=1.1cm]
                \footnotesize
                \fill [c08!40] (0,4) rectangle (8,3);
                \fill [c08!40] (0,3) rectangle (2,0);

                \foreach \i in {0, ..., 4} {
                    \draw (0,\i) -- (8,\i);
                    \draw (2*\i,0) -- (2*\i,4);
                }

                \draw (0,4) -- (2,3);

                \foreach \i/\l in {1/$1$, 2/$2$, 3/$\geq 3$} {
                    \node (c\i) at (1.0, 3.5-\i)   {\l};
                    \node (f\i) at (2*\i+1.0, 3.5) {\l};
                }

                \node (lf) at (1.4, 3.7) {Values};
                \node (lc) at (0.7, 3.3) {Colours};

                \foreach \c/\v/\l in {1/1/$\mathcal{P}$,   1/2/$\mathcal{P}^{~*}$,  1/3/$\mathcal{NPC}$,
                                      2/1/$\mathcal{P}$,   2/2/$\mathcal{P}^{~**}$, 2/3/$\mathcal{NPC}$,
                                      3/1/$\mathcal{NPC}$, 3/2/$\mathcal{NPC}$,     3/3/$\mathcal{NPC}$}
                    \node (r\c\v) at (2*\v+1.0, 3.7-\c) {\l};

                \node (lx) at (3.2, -0.3) {\tiny *  The problem remains open for networks with cycles and
                                                    any two distinct flow values on the arcs};
                \node (lx) at (3.0, -0.6) {\tiny ** Only if each colour is associated to a flow value, and
                                                    one of such values divides the other};

                \node (r11) at (3.0, 2.3) {\resizebox{!}{5pt}{{\color{red}Lemma \ref{lem:ksplitmult}}  \cite{Hartman}}};
                \node (r12) at (5.0, 2.3) {\resizebox{!}{5pt}{{\color{red}Theorems \ref{the:ksplit2v}} \cite{Hartman} \color{red} and \ref{the:mincostcfd2v}}};
                \node (r13) at (7.0, 2.3) {\resizebox{!}{5pt}{\color{red}\cref{the:ksplit3v1c}}};
                \node (r11) at (3.0, 1.3) {\resizebox{!}{5pt}{\color{red}\cref{the:mincostcfduni}}};
                \node (r12) at (5.0, 1.3) {\resizebox{!}{5pt}{\color{red}\cref{the:dfcolcstmin2mult}}};
                \node (r13) at (7.0, 1.3) {\resizebox{!}{5pt}{\color{red}\cref{the:dfcolcstk2c3v}}};
                \node (r11) at (3.0, 0.3) {\resizebox{!}{5pt}{\color{red}\cref{the:kcostdfc3c1v}}};
                \node (r12) at (5.0, 0.3) {\resizebox{!}{5pt}{\color{red}\cref{cor:kcostcfd3c}}};
                \node (r13) at (7.0, 0.3) {\resizebox{!}{5pt}{\color{red}\cref{cor:kcostcfd3c}}};
            \end{tikzpicture}
            \caption{Complexity results for the \textsc{KCostCFD} problem.}
            \label{tbl:results}
        \end{table}
            \FloatBarrier

    \section{Concluding Remarks}

        In this work, we proposed and studied the problem of decomposing a given $(s,t)$-flow $x$ in an arc-coloured network $\Net$ into $(s,t)$-path flows with a minimum cost, where the cost is defined as the sum of the costs of the paths, and the cost of each path is given by its the number of distinct colours. Among the real world applications for this problem, we may mention, for example, telecommunication networks and multimodal transportation systems. The colours may represent risks or different means of transportation.

        We showed that this problem is difficult to solve for networks with a small number of colours, even for uniform flows on networks in general with three colours and on acyclic networks with at least five colours.

        As future works, one should continue investigating the \textsc{MinCostCFD} restricted to the following cases:

        \begin{itemize}
            \item the network has exactly two colours and each colour is associated to a flow value, and the smallest value does not divide the largest;
            \item the network has exactly two colours and two flow values and there is no association between colour and flow value;
            \item uniform flows in acyclic networks with three or four colours.
        \end{itemize}

        Another natural research line when faced to an $\NP$-complete problem is the search for an approximate algorithm with a good approximation factor to the problem.

        We should also investigate three variations of the problem of decomposing a flow $x$ in an arc-coloured network into path flows $x^1, \ldots, x^\ell$, where each $x^i$ is sent along a path $P_i$, with the following objectives:
        \begin{itemize}
            \item minimize $\textstyle \sum_{i=1}^{\ell} \sum_{j \hspace{.5mm} \in \hspace{.5mm} colours(P_i)} span(j,P_i)$;
            \item minimize $\textstyle \sum_{i=1}^{\ell} n_c(P_i)^2$;
            \item maximize $\textstyle \sum_{i=1}^{\ell} \frac{|x^i|}{n_c(P_i)}$.
        \end{itemize}

        In the first approach, the cost of a path is given by the sum of the \emph{span} of each colour in it. Consider, for example, a multimodal transport system. Finding a path flow with fewer colours corresponds to creating a route using fewer types of transportation modes. Additionally, it may not be desirable to switch between modes of transport constantly. The \emph{span} of each colour along the path corresponds to the number of times it will be necessary to take the corresponding transportation mode along that path.

        In the second approach, we want to obtain a decomposition in which the number of colours of the paths are as close as possible. Thinking of the colours as risks, we want paths with almost the same number of risks. For instance, in the original problem we may have two decompositions of a flow into two path flows with a cost of $10$, one that the costs of the paths are $1$ and $9$, and other that the costs are $4$ and $6$. In this case, the second decomposition has a lower cost, since $4^2 + 6^2 = 52 < 1^2 + 9^2 = 82$.

        In the last approach, we take into account both a quantitative and a qualitative aspect (the value and the colours) of each path flow. We are interested in a decomposition that sends large flow values along paths with fewer colours. Once again, thinking of the colours as risks, we want to find safer routes for sending large amount of commodities.

    \bibliographystyle{abbrv}
    \bibliography{references}

\begin{thebibliography}{10}

\bibitem{Ahuja}
R.~K. Ahuja, T.~L. Magnanti, and J.~B. Orlin.
\newblock {\em Network Flows: Theory, Algorithms, and Applications}.
\newblock Prentice hall, New Jersey, 1st edition, 1993.

\bibitem{BaierKSF}
G.~Baier, E.~K{\"{o}}hler, and M.~Skutella.
\newblock The k-splittable flow problem.
\newblock {\em Algorithmica}, 42(3-4), 2005.

\bibitem{BangJensen}
J.~Bang-Jensen and G.~Z. Gutin.
\newblock {\em Digraphs: Theory, Algorithms and Applications}.
\newblock Monographs in Mathematics. Springer Publishing Company, Incorporated,
  New York, 2nd edition, 2008.

\bibitem{BangJensenSubdivision}
J.~Bang-Jensen, F.~Havet, and A.~K. Maia.
\newblock Finding a subdivision of a digraph.
\newblock {\em Theoretical Computer Science}, 562:283--303, 2015.

\bibitem{BondyMurty}
J.~A. Bondy and U.~S.~R. Murty.
\newblock {\em Graph Theory}, volume 244 of {\em Graduate Texts in
  Mathematics}.
\newblock Springer, New York, NY, 1st edition, 2008.

\bibitem{Cormen}
T.~H. Cormen, C.~E. Leiserson, R.~L. Rivest, and C.~Stein.
\newblock {\em Introduction to Algorithms}.
\newblock MIT Press, London, England, 4th edition, 2022.

\bibitem{CoudertSRRG2007}
D.~Coudert, P.~Datta, S.~Perennes, H.~Rivano, and M.-E. Voge.
\newblock Shared risk resource group: Complexity and approximability issues.
\newblock {\em Parallel Processing Letters}, 17(2):169--184, 6 2007.

\bibitem{FordFulkerson}
L.~R. Ford and D.~R. Fulkerson.
\newblock {\em Flows in Networks}.
\newblock Princeton University Press, Princeton, NJ, 1st edition, 1956.

\bibitem{FortuneKLinkage}
S.~Fortune, J.~Hopcroft, and J.~Wyllie.
\newblock The directed subgraph homeomorphism problem.
\newblock {\em Theoretical Computer Science}, 10(2):111--121, 1980.

\bibitem{GareyJohson}
M.~R. Garey and D.~S. Johnson.
\newblock {\em Computers and Intractability: A Guide to the Theory of
  NP-Completeness}.
\newblock Series of Books in the Mathematical Sciences. W. H. Freeman, 1st
  edition, 1979.

\bibitem{Granata2013}
D.~Granata, R.~Cerulli, M.~G. Scutell{\`a}, and A.~Raiconi.
\newblock Maximum flow problems and an {NP-C}omplete variant on edge-labeled
  graphs.
\newblock In P.~M. Pardalos, D.-Z. Du, and R.~L. Graham, editors, {\em Handbook
  of Combinatorial Optimization}, pages 1913--1948. Springer New York, New
  York, NY, 2013.

\bibitem{Hartman}
T.~Hartman, A.~Hassidim, H.~Kaplan, D.~Raz, and M.~Segalov.
\newblock How to split a flow.
\newblock In {\em IEEE INFOCOM 2012}, pages 828--836, 2012.

\bibitem{Kleinberg96}
J.~M. Kleinberg.
\newblock Single-source unsplittable flow.
\newblock In {\em 37th Annual Symposium on Foundations of Computer Science,
  {FOCS} '96}, pages 68--77, Burlington, Vermont, 1996. {IEEE} Computer
  Society.

\bibitem{Schaefer}
T.~J. Schaefer.
\newblock The complexity of satisfiability problems.
\newblock In {\em Proceedings of the Tenth Annual ACM Symposium on Theory of
  Computing}, STOC '78, pages 216--226, New York, NY, USA, 1978. Association
  for Computing Machinery.

\bibitem{Vatinlen2008}
B.~Vatinlen, F.~Chauvet, P.~Chrétienne, and P.~Mahey.
\newblock Simple bounds and greedy algorithms for decomposing a flow into a
  minimal set of paths.
\newblock {\em European Journal of Operational Research}, 185(3):1390--1401,
  2008.

\bibitem{YuanMCPP}
S.~Yuan, S.~Varma, and J.~P. Jue.
\newblock Minimum-color path problems for reliability in mesh networks.
\newblock In {\em Proceedings IEEE 24th Annual Joint Conference of the IEEE
  Computer and Communications Societies.}, volume~4, pages 2658--2669, 2005.

\end{thebibliography}
\end{document}